\documentclass[a4paper]{article}

\usepackage{bbm}
\usepackage{amsmath,amssymb,graphicx}
\usepackage{epsfig}
\usepackage{epic,eepic}
\usepackage{a4wide}

\newcommand{\DS}[1]{{\displaystyle #1}}

\newcommand{\UL}[1]{{\underline #1}}

\newcommand{\nbP}{\mathbbm{P}}

\newcommand{\nbE}{\mathbbm{E}}

\renewcommand{\Pr}{\nbP}
\renewcommand{\phi}{\varphi}
\newcommand{\sleq}{\leq_{\rm st}}

\def\qed{\mbox{\rule[0pt]{1.5ex}{1.5ex}} \vspace{0.2cm}}

\newtheorem{theorem}{Theorem}

\newenvironment{proof}{\noindent \textbf{Proof.}\ }{\qed\par\vskip 4mm\par}

\begin{document}

\title{Optimization results for a generalized coupon collector problem
\thanks{This work was partially funded by the French ANR project SocioPlug (ANR-13-INFR-0003), and by the DeSceNt project granted by the Labex CominLabs excellence laboratory (ANR-10-LABX-07-01).}}

\author{Emmanuelle Anceaume$^1$, Yann Busnel$^{2,4}$, 
Ernst Schulte-Geers$^3$, Bruno Sericola$^4$}

\date{\small
$^1$ CNRS, Campus de Beaulieu, 35042 Rennes Cedex, France\\ 
$^2$ Ensai, Campus de Ker-Lann, BP 37203, 35172 Bruz, Cedex, France\\
$^3$ BSI, Godesberger Allee 185-189, 53175 Bonn, Germany\\
$^4$ Inria, Campus de Beaulieu, 35042 Rennes Cedex, France} 

\maketitle 

\begin{abstract}
We study in this paper a generalized coupon collector problem, which consists
in analyzing the time needed to collect
a given number of distinct coupons that are drawn from a set of coupons
with an arbitrary probability distribution. We suppose that a special coupon
called the null coupon can be drawn but never belongs to any collection.
In this context, we prove that the almost uniform distribution, for which all
the non-null coupons have the same drawing probability, is the distribution
which stochastically minimizes the time
needed to collect a 
fixed number of distinct coupons. 
Moreover, we show that in a given closed subset of probability distributions, the distribution with all its entries, but one, equal to the smallest possible value is the one, which 
stochastically maximizes the time needed to collect a 
fixed number of distinct coupons.
An computer science application shows the utility of these results.

\begin{center}
	\textbf{Keywords}
\end{center} 

Coupon collector problem; Optimization; Schur-convex functions.
\end{abstract}


\section{Introduction}
The coupon collector problem is an old problem, which consists in evaluating
the time needed to get a collection 
of different objects drawn randomly
using a given probability distribution.  
This problem has given rise to a lot of attention
from researchers in various fields since it has applications in many 
scientific domains including computer science and optimization,
see \cite{Boneh97} for several engineering examples.

More formally, consider a set of $n$ coupons, which are drawn randomly one by 
one, with replacement, coupon $i$ being drawn with probability $p_i$.
The classical coupon collector problem is         
to determine the expectation or the distribution of the number 
of coupons that need to be drawn from the set of $n$ coupons
to obtain the full collection of the $n$ coupons.
A large number of papers have been devoted to the analysis of asymptotics
and limit distributions of this distribution when $n$ tends to infinity,
see \cite{Doumas12a} or \cite{Neal08} and the references therein. 
In \cite{Brown08}, the authors obtain new formulas concerning this
distribution and they also provide simulation techniques to compute it as
well as analytic bounds of it.
The asymptotics of the rising moments are studied in \cite{Doumas12b}.

We suppose in this paper that $p=(p_1,\ldots,p_n)$ is not necessarily a probability distribution, \emph{i.e.}, we suppose that
$\sum_{i=1}^n p_i \leq 1$ and we define 
$p_0 = 1 - \sum_{i=1}^n p_i$.
This means that there is a null coupon, denoted by $0$,
which is drawn with probability $p_0$, but that does not belong 
to the collection. 
We are interested, in this setting, in the time needed to collect $c$ different coupons among coupons $1,\ldots,n$, when a coupon is drawn, with replacement, at each discrete time $1,2,\ldots$ among coupons $0,1,\ldots,n$.
This time is denoted by $T_{c,n}(p)$ for $c=1,\ldots,n$. Clearly,
$T_{n,n}(p)$ is the time needed to get the full collection.
The random variable $T_{c,n}(p)$ has been considered in 
\cite{Rubin65} in the case where the drawing probability distribution is 
uniform. The expected value $\nbE[T_{c,n}(p)]$ has been obtained in \cite{Flajolet92} when $p_0=0$. Its distribution and its moments have been obtained in \cite{ABS15} using Markov chains.

In this paper, we prove that the almost uniform distribution, denoted by $v$ and defined by $v=(v_1,\ldots,v_n)$ with $v_i = (1-p_0)/n$, where
$p_0$ is fixed, is the distribution
which stochastically minimizes the time $T_{c,n}(p)$
such that $p_0 = 1 - \sum_{i=1}^n p_i$.
This result was expressed as a conjecture in \cite{ABS15}
where it is proved that the result is true for $c=2$ and for $c=n$ extending the sketch of the proof proposed in 
\cite{Boneh97} to the case $p_0 > 0$. It has been also proved in \cite{ABS15} that the result is true for the expectations, that is that $\nbE[T_{c,n}(u)] \leq \nbE[T_{c,n}(v)] \leq \nbE[T_{c,n}(p)].$

We first consider in Section 2, the case where $p_0 = 0$ and then we extend it to one of $p_0 > 0$.
We show moreover in Section 3, that in a given closed subset of probability distributions, the distribution with all its entries, but one, equal to the smallest possible value is the one which
stochastically maximizes the time $T_{c,n}(p)$.
This work is motivated by the worst case analysis of the behavior of streaming algorithms in network monitoring applications as shown in Section 4.

\section{Distribution minimizing the distribution of $T_{c,n}(p)$}
Recall that $T_{c,n}(p)$, with $p=(p_1,\ldots,p_n)$, is the number 
of coupons that need to
be drawn from the set $\{0,1,2,\ldots,n\}$, with replacement, till one first
obtains a collection with $c$ different coupons, $1 \leq c \leq n$,
among $\{1,\ldots,n\}$, where coupon $i$ is drawn with probability $p_i$,
$i=0,1,\ldots,n$.

The distribution of $T_{c,n}(p)$ has been obtained in \cite{ABS15} using 
Markov chains and is given by
\begin{equation} \label{eqbase}
\Pr\{T_{c,n}(p) > k\} = \sum_{i=0}^{c-1} (-1)^{c-1-i} {n-i-1 \choose n-c}
\sum_{J \in S_{i,n}} (p_0 + P_J)^k,
\end{equation}
where
$S_{i,n} = \left\{J \subseteq \{1,\ldots,n\} \mid |J| = i \right\}$ and,
for every $J \subseteq \{1,\ldots,n\}$, $P_J$ is defined by 
$P_J = \sum_{j \in J} p_j$. Note that we have $S_{0,n} = \emptyset$,  
$P_\emptyset = 0$ and $|S_{i,n}| = {n \choose i}$.

This result also shows as expected that the function
$\Pr\{T_{c,n}(p) > k\}$, as a function of $p$,
is symmetric, which means that it has the same value for any permutation of the entries of $p$.

We recall that if $X$ and $Y$ are two real random variable then
we say that $X$ is stochastically smaller (resp. larger) 
than $Y$, and we write $X \sleq Y$ (resp. $Y \sleq X$),
if $\Pr\{X > t\} \leq \Pr\{Y > t\}$, for all real numbers $t$.
This stochastic order is also referred to as the strong stochastic order.

We consider first, in the following subsection, the case where $p_0=0$.
 
\subsection{The case $p_0=0$} \label{p00}
This case corresponds the fact that there is no null coupon, which means that
all the
coupons can belong to the collection. We thus have
$\sum_{i=1}^n p_i = 1$.
For all $n \geq 1$, $i=1,\ldots,n$ and $k \geq 0$,
we denote by $N_i^{(k)}$ the number of coupons of type $i$
collected at instants $1,\ldots,k$. It is well-known that the joint 
distribution of the $N_i^{(k)}$ is a multinomial distribution, \emph{i.e.},
for all $k_1, \ldots,k_n \geq 0$ such that
$\sum_{i = 1}^n k_i = k$, we have 
\begin{equation} \label{multinome}
\Pr\{N_1^{(k)} = k_1,\ldots,N_n^{(k)} = k_n\} = 
\frac{k!}{k_1!\cdots k_n!} p_1^{k_1} \cdots p_n^{k_n}.
\end{equation}
We also denote by $U_n^{(k)}$ the number of distinct coupon's types, among $1,\ldots,n$, already drawn at instant $k$.
We clearly have, with probability $1$, $U_n^{(0)} = 0$, $U_n^{(1)} = 1$ and, for $i=0,\ldots,n$,
$$\Pr\{U_n^{(k)} = i\} = \sum_{J \in S_{i,n}}
\Pr\{N_u^{(k)} > 0, \; u \in J \mbox{ and } N_u^{(k)} = 0, \; u \notin J\}.$$
Moreover, it is easily checked that, 
$$T_{c,n}(p) > k \Longleftrightarrow U^{(k)} < c.$$
We then have
\begin{eqnarray*}
\Pr\{T_{c,n}(p) > k\} & = & \Pr\{U_n^{(k)} < c\} \\
& = & \sum_{i=0}^{c-1} \Pr\{U_n^{(k)} = i\} \\
& = & \sum_{i=0}^{c-1} \sum_{J \in S_{i,n}}
\Pr\{N_u^{(k)} > 0, \; u \in J \mbox{ and } N_u^{(k)} = 0, \; u \notin J\}.
\end{eqnarray*}
Using Relation (\ref{multinome}), we obtain
\begin{equation} \label{vide}
\Pr\{T_{c,n}(p) > k\} = \sum_{i=0}^{c-1} \sum_{J \in S_{i,n}}
\sum_{\UL{k} \in E_{k,J}}
k! \prod_{j \in J} \frac{p_j^{k_j}}{k_j!},
\end{equation}
where $E_{k,J}$ is the set of vectors defined by
$$E_{k,J} = \left\{\UL{k} = (k_j)_{j \in J} \mid k_j > 0, 
\mbox{ for all } j \in J
\mbox{ and } K_J = k\right\},$$
with $K_J = \sum_{j \in J} k_j$.
 
\begin{theorem} \label{autreloi}
For all $n \geq 2$ and $p=(p_1,\ldots,p_n) \in (0,1)^n$ with
$\sum_{i=1}^n p_i = 1$, and for all $c=1,\ldots,n$,
we have $T_{c,n}(p') \sleq T_{c,n}(p)$,
where $p'=(p_1,\ldots,p_{n-2},p'_{n-1},p'_n)$
with
$p'_{n-1} = \lambda p_{n-1} + (1-\lambda) p_n$ and
$p'_{n} = (1-\lambda) p_{n-1} + \lambda p_n$,
for all $\lambda \in [0,1]$. 
\end{theorem}

\begin{proof}
The result is trivial for $c=1$, since we have
$T_{1,n}(p) = 1$, and for $k = 0$ since both terms are equal to $1$.
We thus suppose now that $c \geq 2$ and $k \geq 1$.
The fact that $k \geq 1$ means in particular that the term $i=0$ in
Relation (\ref{vide}) is equal to $0$. We then have
\begin{equation} \label{vide1}
\Pr\{T_{c,n}(p) > k\} = \sum_{i=1}^{c-1} \sum_{J \in S_{i,n}}
\sum_{\UL{k} \in E_{k,J}}
k! \prod_{j \in J} \frac{p_j^{k_j}}{k_j!}.
\end{equation}

To simplify the notation, we denote by $T_i(p)$ the $i$th term of this sum,
that is
\begin{equation} \label{Ti}
T_i(p) = \sum_{J \in S_{i,n}} \sum_{\UL{k} \in E_{k,J}}
k! \prod_{j \in J} \frac{p_j^{k_j}}{k_j!}.
\end{equation}

For $i=1$, we have $S_{1,n} = \{\{1\},\ldots,\{n\}\}$ and 
$E_{k,\{j\}} = \{k\}$. The term $T_1(p)$ is thus given by
$$T_1(p) = \sum_{j=1}^{n} p_j^k.$$
Consider now the term $T_i$ for $i \geq 2$.
We split the set $S_{i,n}$ into four subsets depending on wether the indices
$n-1$ and $n$ belong to its elements. More precisely, we introduce the 
partition of the set $S_{i,n}$ in the four subsets 
$S_{i,n}^{(1)}$, $S_{i,n}^{(2)}$, $S_{i,n}^{(3)}$ and $S_{i,n}^{(4)}$
defined by
\begin{eqnarray*}
S_{i,n}^{(1)} & = & \left\{J \subseteq \{1,\ldots,n\} \mid |J| = i 
                   \mbox{ with } n-1 \in J \mbox{ and } n \notin J\right\}, \\
S_{i,n}^{(2)} & = & \left\{J \subseteq \{1,\ldots,n\} \mid |J| = i 
                   \mbox{ with } n-1 \notin J \mbox{ and } n \in J\right\}, \\
S_{i,n}^{(3)} & = & \left\{J \subseteq \{1,\ldots,n\} \mid |J| = i 
                   \mbox{ with } n-1 \in J \mbox{ and } n \in J\right\}, \\
S_{i,n}^{(4)} & = & \left\{J \subseteq \{1,\ldots,n\} \mid |J| = i 
                   \mbox{ with } n-1 \notin J \mbox{ and } n \notin J\right\}.
\end{eqnarray*}
These subsets can also be written as
$S_{i,n}^{(1)} = S_{i-1,n-2} \cup \{n-1\}$, 
$S_{i,n}^{(2)} = S_{i-1,n-2} \cup \{n\}$, 
$S_{i,n}^{(3)} = S_{i-2,n-2} \cup \{n-1,n\}$, 
and $S_{i,n}^{(4)} = S_{i,n-2}$.
Using these equalities, the term $T_i$ of Relation (\ref{Ti}) becomes
\begin{eqnarray*}
T_i(p) 
& = & \sum_{J \in S_{i-1,n-2}} \sum_{\UL{k} \in E_{k,J\cup \{n-1\}}}
k! \left(\prod_{j \in J} \frac{p_j^{k_j}}{k_j!}\right) 
\frac{p_{n-1}^{k_{n-1}}}{k_{n-1}!} \\
&   & + \; \sum_{J \in S_{i-1,n-2}} \sum_{\UL{k} \in E_{k,J \cup \{n\}}}
k! \left(\prod_{j \in J} \frac{p_j^{k_j}}{k_j!}\right) 
\frac{p_{n}^{k_{n}}}{k_{n}!} \\
&   & + \; \sum_{J \in S_{i-2,n-2}} \sum_{\UL{k} \in E_{k,J \cup \{n-1,n\}}}
k! \left(\prod_{j \in J} \frac{p_j^{k_j}}{k_j!}\right) 
\frac{p_{n-1}^{k_{n-1}}}{k_{n-1}!} \frac{p_{n}^{k_{n}}}{k_{n}!} \\
&   & + \; \sum_{J \in S_{i,n-2}} \sum_{\UL{k} \in E_{k,J}}
k! \left(\prod_{j \in J} \frac{p_j^{k_j}}{k_j!}\right). 
\end{eqnarray*}
Let us now introduce the sets $L_{k,J}$ defined by
$$L_{k,J} = \left\{\UL{k} = (k_j)_{j \in J} \mid k_j > 0, 
\mbox{ for all } j \in J
\mbox{ and } K_J \leq k\right\}.$$
Using these sets and, to clarify the notation, setting $k_{n-1}=\ell$
and $k_{n}=h$ when needed, we obtain
\begin{eqnarray*}
T_i(p) 
& = & \sum_{J \in S_{i-1,n-2}} \sum_{\UL{k} \in L_{k-1,J}}
k! \left(\prod_{j \in J} \frac{p_j^{k_j}}{k_j!}\right) 
\frac{p_{n-1}^{k-K_J}}{(k-K_J)!} \\
&   & + \; \sum_{J \in S_{i-1,n-2}} \sum_{\UL{k} \in L_{k-1,J}}
k! \left(\prod_{j \in J} \frac{p_j^{k_j}}{k_j!}\right) 
\frac{p_{n}^{k-K_J}}{(k-K_J)!} \\
&   & + \; \sum_{J \in S_{i-2,n-2}} \sum_{\UL{k} \in L_{k-2,J}}
k! \left(\prod_{j \in J} \frac{p_j^{k_j}}{k_j!}\right) 
\sum_{\ell>0,h>0,\ell+h=k-K_J}\frac{p_{n-1}^{\ell}}{\ell!} \frac{p_{n}^{h}}{h!} \\
&   & + \; \sum_{J \in S_{i,n-2}} \sum_{\UL{k} \in E_{k,J}}
k! \left(\prod_{j \in J} \frac{p_j^{k_j}}{k_j!}\right), 
\end{eqnarray*}
which can also be written as
\begin{eqnarray*}
T_i(p) 
& = & \sum_{J \in S_{i-1,n-2}} \sum_{\UL{k} \in L_{k-1,J}}
\frac{k!}{(k-K_J)!} \left(\prod_{j \in J} \frac{p_j^{k_j}}{k_j!}\right) 
\left(p_{n-1}^{k-K_J} + p_{n}^{k-K_J}\right) \\
&   & + \; \sum_{J \in S_{i-2,n-2}} \sum_{\UL{k} \in L_{k-2,J}}
\frac{k!}{(k-K_J)!}\left(\prod_{j \in J} \frac{p_j^{k_j}}{k_j!}\right) 
\left(p_{n-1} + p_{n}\right)^{k-K_J} \\
&   & - \; \sum_{J \in S_{i-2,n-2}} \sum_{\UL{k} \in L_{k-2,J}}
\frac{k!}{(k-K_J)!}\left(\prod_{j \in J} \frac{p_j^{k_j}}{k_j!}\right) 
\left(p_{n-1}^{k-K_J} + p_{n}^{k-K_J}\right) \\
&   & + \; \sum_{J \in S_{i,n-2}} \sum_{\UL{k} \in E_{k,J}}
k! \left(\prod_{j \in J} \frac{p_j^{k_j}}{k_j!}\right).
\end{eqnarray*}
We denote these four terms respectively by $A_i(p)$, $B_i(p)$, $C_i(p)$ 
and $D_i(p)$.
We thus have, for $i \geq 2$, $T_i(p) = A_i(p) + B_i(p) - C_i(p) + D_i(p)$.
We have already shown that $T_1(p) = A_1(p) + D_1(p)$, so we set
$B_1(p) = C_1(p) = 0$.
We then have
\begin{eqnarray}
\Pr\{T_{c,n}(p) > k\} 
& = & \sum_{i=1}^{c-1} T_i(p) \nonumber \\
& = & A_{c-1}(p) +
\sum_{i=1}^{c-2} (A_i(p) - C_{i+1}(p)) + \sum_{i=2}^{c-1} B_i(p) +
\sum_{i=1}^{c-1} D_i(p). \label{ABCD}
\end{eqnarray}
For $i \geq 1$, we have
\begin{eqnarray*}
A_i(p) - C_{i+1}(p) 
& = & \sum_{J \in S_{i-1,n-2}} \sum_{\UL{k} \in L_{k-1,J}}
\frac{k!}{(k-K_J)!} \left(\prod_{j \in J} \frac{p_j^{k_j}}{k_j!}\right) 
\left(p_{n-1}^{k-K_J} + p_{n}^{k-K_J}\right) \\
&   &- \; \sum_{J \in S_{i-1,n-2}} \sum_{\UL{k} \in L_{k-2,J}}
\frac{k!}{(k-K_J)!}\left(\prod_{j \in J} \frac{p_j^{k_j}}{k_j!}\right) 
\left(p_{n-1}^{k-K_J} + p_{n}^{k-K_J}\right) \\
& = & \sum_{J \in S_{i-1,n-2}} \sum_{\UL{k} \in E_{k-1,J}}
\frac{k!}{(k-K_J)!} \left(\prod_{j \in J} \frac{p_j^{k_j}}{k_j!}\right) 
\left(p_{n-1}^{k-K_J} + p_{n}^{k-K_J}\right). 
\end{eqnarray*}
By definition of the set $E_{k-1,J}$, we have $K_J = k-1$ in the previous equality. This gives
$$A_i(p) - C_{i+1}(p) = \sum_{J \in S_{i-1,n-2}} \sum_{\UL{k} \in E_{k-1,J}}
k! \left(\prod_{j \in J} \frac{p_j^{k_j}}{k_j!}\right) 
\left(p_{n-1} + p_{n}\right).$$
Using the fact that
the function $x\longmapsto x^s$ is convex on interval $[0,1]$ for every non negative 
integer $s$, we have
\begin{eqnarray*}
{p'}_{n-1}^{k-K_J} + {p'}_{n}^{k-K_J} 
& = & \left(\lambda p_{n-1} + (1-\lambda) p_n\right)^{k-K_J}
+ \left((1-\lambda) p_{n-1} + \lambda p_n\right)^{k-K_J} \\
& \leq & \lambda p_{n-1}^{k-K_J} + (1-\lambda) p_n^{k-K_J}
+ (1-\lambda) p_{n-1}^{k-K_J} + \lambda p_n^{k-K_J} \\
& = & p_{n-1}^{k-K_J} + p_n^{k-K_J},
\end{eqnarray*}
and in particular $p'_{n-1} + p'_{n} = p_{n-1} + p_{n}$.
It follows that 
\begin{eqnarray*}
& & A_{c-1}(p') \leq A_{c-1}(p), \\ 
& & A_i(p') - C_{i+1}(p') = A_i(p) - C_{i+1}(p), \\
& & B_i(p') = B_i(p), \\
& & D_i(p') = D_i(p),
\end{eqnarray*}
and from (\ref{ABCD}) that
$\Pr\{T_{c,n}(p') > k\} \leq \Pr\{T_{c,n}(p) > k\}$, which concludes the proof.
\end{proof}

The function $\Pr\{T_{c,n}(p) > k\}$, as a function of $p$, being 
symmetric, this theorem can easily 
be extended to the case where the two entries $p_{n-1}$ and $p_n$ of $p$ 
are any $p_i, p_j \in \{p_1,\ldots,p_n\}$, with $i \neq j$.
 
In fact, we have shown in this theorem that for fixed $n$ and $k$, 
the function of $p$, $\Pr\{T_{c,n}(p) \leq k\}$, is a Schur-convex function,
that is, a function that preserves the order of majorization.
See \cite{Marshall81} for more details on this subject.

\begin{theorem} \label{notetheo1}
For every $n \geq 1$ and $p=(p_1,\ldots,p_n) \in (0,1)^n$ with
$\sum_{i=1}^n p_i = 1$, and for all $c=1,\ldots,n$,
we have $T_{c,n}(u) \sleq T_{c,n}(p)$,
where $u=(1/n,\ldots,1/n)$ is the uniform distribution.
\end{theorem}
 
\begin{proof}
To prove this result, we apply successively and at most $n-1$ 
times Theorem \ref{autreloi}
as follows. We first choose two different entries 
of $p$, say $p_i$ and $p_j$ such that $p_i < 1/n < p_j$ and next to
define $p'_i$ and $p'_j$ by
$$p'_i = \frac{1}{n} \mbox{ and } p'_j = p_i + p_j - \frac{1}{n}.$$
This leads us to write
$p'_i = \lambda p_i + (1-\lambda) p_j$ and
$p'_j = (1-\lambda) p_i + \lambda p_j$, with  
$$\lambda = \frac{\DS{p_j - 1/n}}{p_j - p_i}.$$
From Theorem \ref{autreloi}, vector $p'$ obtained by taking the 
other entries equal to those of 
$p$, \emph{i.e.}, by taking $p'_\ell = p_\ell$, 
for $\ell \neq i,j$, 
is such that
$\Pr\{T_{c,n}(p') > k\} \leq \Pr\{T_{c,n}(p) > k\}$.
Note that at this point vector $p'$ has at least one entry equal to
$1/n$, so repeating at most $n-1$ this procedure, we get vector $u$,
which concludes the proof.
\end{proof}

To illustrate the steps used in the proof of this theorem, 
we take the following example. Suppose that $n=5$ and
$p = (1/16,1/6,1/4,1/8,19/48)$. 
In a first step, taking $i=4$ and $j=5$, we get
$$p^{(1)}=(1/16,1/6,1/4,1/5,77/240).$$
In a second step, taking $i=2$ and $j=5$, we get
$$p^{(2)}=(1/16,1/5,1/4,1/5,69/240).$$
In a third step, taking $i=1$ and $j=3$, we get
$$p^{(3)}=(1/5,1/5,9/80,1/5,69/240).$$
For the fourth and last step, taking $i=3$ and $j=5$, we get
$$p^{(4)}=(1/5,1/5,1/5,1/5,1/5).$$

\subsection{The case $p_0>0$}
We consider now the case where $p_0 > 0$.
We have $p=(p_1,\ldots,p_n)$ with $p_1 + \cdots +p_n < 1$ and $p_0= 1-(p_1 + \cdots +p_n)$.
Recall that in this case $T_{c,n}(p)$ is the time or the number of steps needed to collect a subset of $c$ 
different coupons
among coupons $1,\ldots,n$. Coupon $0$ is not allowed to belong to the collection. As in the previous subsection we denote, for $i=0,1,\ldots,n$ and $k \geq 0$, by $N_i^{(k)}$ the number of coupons of type $i$ collected at instants $1,\ldots,k$.
We then have
for all $k_0,k_1, \ldots,k_n \geq 0$ such that
$\sum_{i = 0}^n k_i = k$, we have 
$$\Pr\{N_0^{(k)}=k_0,N_1^{(k)} = k_1,\ldots,N_n^{(k)} = k_n\} = 
\frac{k!}{k_0!k_1!\cdots k_n!} p_0^{k_0}p_1^{k_1} \cdots p_n^{k_n},$$
which can also be written as
\begin{align}
\Pr\{N_0^{(k)}= & k_0,N_1^{(k)} = k_1,\ldots,N_n^{(k)} = k_n\} \nonumber \\
& = {k \choose k_0} p_0^{k_0} (1-p_0)^{k-k_0}
\frac{(k-k_0)!}{k_1!\cdots k_n!} \left(\frac{p_1}{1-p_0}\right)^{k_1} \cdots 
\left(\frac{p_n}{1-p_0}\right)^{k_n}. \label{mul1}
\end{align}
Note that $p/(1-p_0)$ is a probability distribution since
$$\frac{1}{1-p_0} \sum_{i=1}^n p_i = 1,$$
so summing over all the $k_1, \ldots,k_n \geq 0$ such that
$\sum_{i = 1}^n k_i = k-k_0$, we get, for all $k_0 = 0,\ldots,k$,
\begin{equation} \label{mul2}
\Pr\{N_0^{(k)}=k_0\} = {k \choose k_0} p_0^{k_0} (1-p_0)^{k-k_0}.
\end{equation}
From (\ref{mul1}) and (\ref{mul2}) we obtain,
for all $k_1, \ldots,k_n \geq 0$ such that
$\sum_{i = 1}^n k_i = k-k_0$,
\begin{equation} \label{mul3}
\Pr\{N_1^{(k)} = k_1,\ldots,N_n^{(k)} = k_n \mid N_0^{(k)}=k_0\} 
= \frac{(k-k_0)!}{k_1!\cdots k_n!} \left(\frac{p_1}{1-p_0}\right)^{k_1} \cdots 
\left(\frac{p_n}{1-p_0}\right)^{k_n}.
\end{equation}
Recall that $U_n^{(k)}$ is the number of coupon's types 
among $1,\ldots,n$ already drawn at instant $k$.
We clearly have, with probability $1$, $U_n^{(0)} = 0$, and,
for $i=0,\ldots,n$,
$$\Pr\{U_n^{(k)} = i\} = \sum_{J \in S_{i,n}}
\Pr\{N_u^{(k)} > 0, \; u \in J \mbox{ and } N_u^{(k)} = 0, \; u \notin J\}.$$
Moreover, we have as in Subsection \ref{p00}, 
$$T_{c,n}(p) > k \Longleftrightarrow U_n^{(k)} < c.$$
and so
\begin{align*}
\Pr\{T_{c,n}(p) > k \mid & N_0^{(k)}=k_0\} \\
& = \sum_{i=0}^{c-1} \sum_{J \in S_{i,n}}
\Pr\{N_u^{(k)} > 0, \; u \in J \mbox{ and } N_u^{(k)} = 0, \; u \notin J \mid N_0^{(k)}=k_0\}.
\end{align*}
Using Relation (\ref{mul3}), we obtain
\begin{equation} \label{vide2}
\Pr\{T_{c,n}(p) > k \mid N_0^{(k)}=k_0\} = \sum_{i=0}^{c-1} \sum_{J \in S_{i,n}}
\sum_{\UL{k} \in E_{k-k_0,J}}
(k-k_0)! \left(\prod_{j \in J} \frac{\left(\frac{p_j}{1-p_0}\right)^{k_j}}{k_j!}\right),
\end{equation}
where the set $E_{k,J}$ has been defined in Subsection 2.1. 

\begin{theorem} \label{notetheo2}
For every $n \geq 1$ and $p=(p_1,\ldots,p_n) \in (0,1)^n$ with
$\sum_{i=1}^n p_i < 1$, and for all $c=1,\ldots,n$,
we have
$T_{c,n}(u) \sleq T_{c,n}(v) \sleq T_{c,n}(p)$,
where $u=(1/n,\ldots,1/n)$, $v=(v_1,\ldots,v_n)$ with $v_i =(1-p_0)/n$ and 
$p_0 =1 - \sum_{i=1}^n p_i$.
\end{theorem}
 
\begin{proof}
From Relation (\ref{vide}) and Relation (\ref{vide2}), we obtain, for all
$k_0 = 0,\ldots,k$,
\begin{equation} \label{min0}
\Pr\{T_{c,n}(p) > k \mid N_0^{(k)}=k_0\} = \Pr\{T_{c,n}(p/(1-p_0)) > k-k_0\}.
\end{equation}
Using (\ref{mul2}) and unconditioning, we obtain
\begin{equation} \label{min}
\Pr\{T_{c,n}(p) > k\} = \sum_{\ell=0}^{k} {k \choose \ell} p_0^\ell
(1-p_0)^{k - \ell}
\Pr\{T_{c,n}(p/(1-p_0)) > k-\ell\}.
\end{equation}
Since $p/(1-p_0)$ is a probability distribution, 
applying Theorem \ref{notetheo1} to this distribution 
and observing that $u = v/(1-p_0)$, we get
\begin{eqnarray*}
\Pr\{T_{c,n}(p) > k\} 
& = & \sum_{\ell=0}^{k} {k \choose \ell} 
p_0^\ell (1-p_0)^{k - \ell} \Pr\{T_{c,n}(p/(1-p_0)) > k-\ell\} \\
& \geq & \sum_{\ell=0}^{k} {k \choose \ell} 
p_0^\ell (1-p_0)^{k - \ell} \Pr\{T_{c,n}(u) > k-\ell\} \\
& = & \sum_{\ell=0}^{k} {k \choose \ell} 
p_0^\ell (1-p_0)^{k - \ell} \Pr\{T_{c,n}(v/(1-p_0)) > k-\ell\} \\
& = & \Pr\{T_{c,n}(v) > k\},
\end{eqnarray*}
where the last equality follows from (\ref{min}).
This proves the second inequality.

To prove the first inequality, observe that 
$\Pr\{T_{c,n}(p/(1-p_0)) > \ell\}$ is decreasing with $\ell$.
This leads, using (\ref{min}), to
$$\Pr\{T_{c,n}(p) > k\} \geq \Pr\{T_{c,n}(p/(1-p_0)) > k\}.$$
Taking $p=v$ in this inequality gives
$$\Pr\{T_{c,n}(v) > k\} \geq \Pr\{T_{c,n}(u) > k\},$$
which completes the proof.
\end{proof}

In fact, we have shown in the proof of this theorem and more precisely in Relation (\ref{min0}) and using Theorem \ref{notetheo1} that for fixed $n,k$ and $k_0$, 
the function of $p$, $\Pr\{T_{c,n}(p) \leq k \mid N_0^{(k)}=k_0\}$ is a Schur-convex function, that is, a function that preserves the order of majorization.
In particular, from (\ref{min}), $\Pr\{T_{c,n}(p) \leq k\}$ is also a Schur-convex function,
even when $p_0 > 0$.
See \cite{Marshall81} for more details on this subject.

\section{Distribution maximizing the distribution of $T_{c,n}(p)$}

We consider in this section the problem of stochastically maximizing the time $T_{c,n}(p)$ when $p$ varies in a closed subset of dimension $n$. In the previous section the
minimization was made on the set ${\cal A}$ defined,
for every $n$ and for all $p_0 \in [0,1)$, by
$${\cal A} = \{p=(p_1,\ldots,p_n) \in (0,1)^n \mid 
p_1 + \cdots + p_n = 1 -p_0\}.$$
According to the application described in the Section 4,
we fix a parameter $\theta \in (0, (1-p_0)/n]$ and we are looking for distributions
$p$ which stochastically maximizes the time $T_{c,n}(p)$ on the subset ${\cal A}_\theta$ of ${\cal A}$
defined by
$${\cal A}_\theta =
\{p=(p_1,\ldots,p_n) \in  {\cal A} \mid 
p_j \geq \theta, \mbox{ for every } j = 1,\ldots,n\}.$$
The solution to this problem is given by the following theorem.
We first introduce the set ${\cal B}_\theta$ defined by the distributions of 
${\cal A}_\theta$ with all their entries, except one, are equal to $\theta$.
The set ${\cal B}_\theta$ has $n$ elements given by
$${\cal B}_\theta = \{(\gamma,\theta,\ldots,\theta),
(\theta,\gamma,\theta,\ldots,\theta),\ldots,
(\theta,\ldots,\theta,\gamma)\},$$
where $\gamma = 1 -p_0- (n-1)\theta$.
Note that since $\theta \in (0,(1-p_0)/n]$, we have 
$1-p_0-(n-1)\theta \geq \theta$ which means that 
${\cal B}_\theta \subseteq {\cal A}_\theta$.

\begin{theorem} \label{notetheo3}
For every $n \geq 1$ and $p=(p_1,\ldots,p_n) \in {\cal A}_\theta$ 
and for all $c=1,\ldots,n$,
we have $T_{c,n}(p) \sleq T_{c,n}(q)$,
for every $q \in {\cal B}_\theta$.
\end{theorem}
 
\begin{proof}
Since $\Pr\{T_{c,n}(p) > k\}$ is a symmetric function of $p$, 
$\Pr\{T_{c,n}(q) > k\}$ has the same value for every
$q \in {\cal B}_\theta$, 
so we suppose that $q_\ell=\theta$ for every $\ell \neq j$ and $q_j=\gamma$.

Let $p \in {\cal A}_\theta \setminus {\cal B}_\theta$ and let $i$ be the first entry of $p$ such that $i \neq j$ and $p_i > \theta$.
We then define the distribution $p^{(1)}$ as
$p_i^{(1)}=\theta$, $p_j^{(1)}=p_i+p_j-\theta > p_j$ and
$p_\ell^{(1)} = p_\ell$, for $\ell \neq i,j$.
This leads us to write 
$p_i = \lambda p_i^{(1)} + (1- \lambda) p_j^{(1)}$
and
$p_j = (1-\lambda) p_i^{(1)} + \lambda p_j^{(1)}$,
with
$$\lambda = \frac{p_j^{(1)} - p_i}{p_j^{(1)} - p_i^{(1)}}
= \frac{p_j - \theta}{p_j - \theta + p_i - \theta} \in [0,1).$$
From Theorem \ref{autreloi}, we get
$\Pr\{T_{c,n}(p) > k\} \leq \Pr\{T_{c,n}(p^{(1)}) > k\}$.
Repeating the same procedure from distribution $p^{(1)}$ and so on, we get, after at most $n-1$ steps, distribution $q$, that is
$\Pr\{T_{c,n}(p) > k\} \leq \Pr\{T_{c,n}(q) > k\}$, 
which completes the proof.
\end{proof}

To illustrate the steps used in the proof of this theorem, 
we take the following example. Suppose that $n=5$,
$\theta = 1/20$, $p_0=1/10$ and $q =(1/20,1/20,1/20,7/10,1/20)$, which means that $j=4$.
Suppose moreover that $p=(1/16,1/6,1/4,1/8,71/240)$.
We have $p \in {\cal A}_\theta$ and $q \in {\cal B}_\theta$.
In a first step, taking $i=1$ and since $j=4$, we get
$$p^{(1)}=(1/20,1/6,1/4,11/80,71/240).$$
In a second step, taking $i=2$ and since $j=4$, we get
$$p^{(2)}=(1/20,1/20,1/4,61/240,71/240).$$
In a third step, taking $i=3$ and since $j=4$, we get
$$p^{(3)}=(1/20,1/20,1/20,109/240,71/240).$$
For the fourth and last step, taking $i=5$ and since $j=4$, we get
$$p^{(4)}=(1/20,1/20,1/20,7/10,1/20)=q.$$

\section{Application to the detection of distributed denial of service attacks}
\label{sec:application}

A Denial of Service (DoS) attack tries to  progressively take  down an Internet  
resource by flooding this resource with more requests than it is capable to 
handle. A Distributed Denial of Service (DDoS) attack is a DoS attack triggered 
by thousands of machines that have been infected by a malicious software, 
with as immediate consequence the total shut down of targeted web resources 
(\emph{e.g.}, e-commerce websites).
A solution to detect and  to mitigate DDoS attacks is to monitor network 
traffic at routers and to look for highly frequent signatures that might 
suggest  ongoing attacks. 
A recent strategy followed by the attackers is to hide their massive flow of 
requests over a multitude of routes, so that locally, these requests do not 
appear as frequent, while globally each of these requests represent a portion $\theta$ of 
the network traffic. The term ``iceberg'' has been recently introduced~\cite{ZLOX10} to describe such an attack as only a very small part of the iceberg can be observed at each single router. The approach adopted to defend against such new attacks is to rely on multiple routers that locally  monitor their network traffic, and  upon detection of potential icebergs,  inform a monitoring server that aggregates all the monitored information  to accurately detect icebergs.  Now to prevent the server from being overloaded by all the monitored information, routers continuously keep track of the $c$ (among $n$) most recent distinct requests. These requests  locally represent at least a fraction $\theta$ of the local stream. Once collected, each  router  sends them to the server, and throw away all the requests $i$ that appear with a  probability $p_i$ smaller than  $\theta$. The sum of these small probabilities is represented by probability $p_0$. Parameter $c$ is dimensioned so that the frequency at which all the routers send their $c$ last requests  is low enough to enable the server  to aggregate  all of them and to trigger a DDoS alarm when needed.  This  amounts to compute the distribution of the time $T_{c,n}(p)$ needed to collect $c$ distinct requests among  the $n$ ones. Theorem~\ref{notetheo2} shows that the distribution $p$ that stochastically minimizes the time $T_{c,n}(p)$ is the almost uniform distribution $v$. This means that if locally each router receives a stream in which all the frequent requests  (that is, those whose probability of occurrence is greater than or equal to $\theta$) occur with same probability, then the complementary distribution of  the time needed to locally complete a collection of $c$ distinct requests is minimized. As a consequence the delay between any two interactions between a router and the server is minimized.

Another important aspect of DDoS detection applications is their ability to bound the detection latency of global icebergs, that is the maximal time that elapses between the presence of a global iceberg at some of the routers and its detection at the server. This can be implemented through a timer that will fire if no communication has been triggered between a router and the server, which happens  if locally a router has received less than $c$ requests. Dimensioning such a timer amounts to determine  the distribution of the maximal time it takes for a router to collect $c$ distinct requests. Theorem~\ref{notetheo3} shows that the distributions $p$ that stochastically maximizes this time $T_{c,n}(p)$ are all  the distributions $q \in {\cal B}_\theta$.

\end{document}